\documentclass[a4paper,12pt]{article}
\usepackage{latexsym,amssymb,amsfonts,amsmath,amsthm}
\usepackage[dvips]{graphicx}
\usepackage{comment}

\newtheorem{theorem}{Theorem}
\newtheorem{lemma}{Lemma}
\newtheorem{corollary}{Corollary}
\newtheorem{proposition}{Proposition}
\newtheorem{remark}{Remark}

\numberwithin{theorem}{section}
\numberwithin{lemma}{section}
\numberwithin{corollary}{section}
\numberwithin{proposition}{section}
\numberwithin{remark}{section}

\usepackage{color}
\newcommand{\red}{\textcolor{black}}

\title{{\Large {\bf A discontinuity of the energy of quantum walk in impurities }}
\author{
{\small Kenta Higuchi }\\
{\scriptsize Department of Mathematical Sciences, Ritsumeikan University,}\\
{\scriptsize Noji-Higashi, Kusatsu, 525-8577, Japan.
}\\
{\small Takashi Komatsu}\\
{\scriptsize Math. Research Institute Calc for Industry
Minami, }\\  
{\scriptsize Hiroshima, 732-0816, Japan. }\\
{\small Norio Konno}\\
{\scriptsize Department of Applied Mathematics, Yokohama National University, }\\  
{\scriptsize Hodogaya, Yokohama 240-8501, Japan.}\\
{\small Hisashi Morioka}\\
{\scriptsize Graduate School of Science and Engineering, Ehime University, }\\  
{\scriptsize Bunkyo-cho 3, Matsuyama, Ehime, 790-8577, Japan.}\\
{\small Etsuo Segawa}\\
{\scriptsize Graduate School of Environment Information Sciences, Yokohama National University, } \\
{\scriptsize Hodogaya, Yokohama 240-8501, Japan.}}
}
\date{\empty }
\begin{document}
\maketitle

\par\noindent
\begin{small}
\par\noindent
{\bf Abstract}. 
We consider the discrete-time quantum walk whose local dynamics is denoted by $C$ at the perturbed region $\{0,1,\dots,M-1\}$ and free  at the other positions. We obtain the stationary state with a bounded initial state. The initial state is set so that the perturbed region receives the inflow $\omega^n$ at time $n$ $(|\omega|=1)$. From this expression, we compute the scattering on the surface of $-1$ and $M$ and also compute the quantity how quantum walker accumulates in the perturbed region; namely the energ of the quantum walk, in the long time limit. We find a discontinuity of the energy with respect to the frequency of the inflow.

\footnote[0]{
{\it Keywords: } 
Quantum walk, Scattering theory, Scattering matrix 
}
\end{small}

\section{Introduction}
There is no doubt that studies on scattering theory is one of the interesting topic of the Schr{\"o}dinger equation.  
Recently, it is revealed that the scatterings of some fundamental stationary Schr{\"o}dinger equations  on the real line with not only delta potentials~\cite{T,HKSS1,MMOS} but also continuous potential~\cite{H} can be recovered by discrete-time quantum walks.  
These induced quantum walks are given by the following setting: the non-trivial quantum coins are assigned to some vertices in a finite 
region on the one-dimensional lattice as the impurities and the free-quantum coins are assigned at the other vertices. The initial state is given so that quantum walkers inflows into the perturbed region at every time step.  
It is shown that the scattering matrix of the quantum walk on the one-dimensional lattice can be explicitly described by using a path counting in \cite{KKMS} and this path counting method can be described by a discrete analogue of the Feynmann path integral~\cite{H}.  
There are some studies for the scattering theory of quantum walks under slightly general settings and related topics \cite{Su,RST1,RST2,Morioka,MoSe,Ti,KKMS2}.

Such a setting is the special setting of \cite{FelHil1,HS} in that  the regions where a quantum walker moves freely coincide with tails in \cite{FelHil1,HS}, and the perturbed region can be regarded as a finite and connected graph in \cite{FelHil1,HS}. 
The properties of not only the scattering on the surface of the internal graph but also the stationary state in the internal graph for the Szegedy walk are characterized by \cite{MHS} with a constant inflow from the tails. 

By \cite{HS}, this quantum walk converges to a stationary state.  
So let $\vec{\varphi}(\cdot): \mathbb{Z}\to \mathbb{C}^2$ be the stationary state of the quantum walk on $\mathbb{Z}$. The perturbed region is $[M]:=\{0,1,\dots,M-1\}$ and we assign the quantum coin 
\[C=\begin{bmatrix}a & b \\ c & d \end{bmatrix}\]
to each vertex in $[M]$. 
The inflow into the perturbed region at time $n$ is expressed by $\omega^n$ $(|\omega|=1)$. 
In this paper, we compute (1) the scattering on the surface of the perturbed region $[M]$ in the one-dimensional lattice; (2) the energy of the quantum walk. Here the energy of quantum walk is defined by 
\[ \mathcal{E}_M(\omega)=\sum_{x=0}^{M-1}||\vec{\varphi}(x)||^2_{\mathbb{C}^2}. \]
This is the quantity how much quantum walkers accumulate to the perturbed region $[M]$ in the long time limit. 
We obtain a necessary and sufficient condition for the perfect transmitting, and also obtain the energy. As a consequence of our result on the energy, we observe a discontinuity of the energy with respect to the frequency of the inflow. Moreover our result implies that the condition for $\theta(\omega)\in \mathbb{N}$ is equivalent to the condition for the perfect transmitting. 
Then we obtain that the situation of the perfect transmitting not only releases quantum walker to the opposite outside but also  accumulates quantum walkers in the perturbed region.  
Note that since this quantum walk can be converted to a quantum walk with absorption walls, the problem is reduced to analysis on a finite matrix $E_M$ which is obtained by picking up from the total unitary time evolution operator with respect to the perturbed region $[M]$. See \cite{Parker} for a precise spectral results on $E_{M}$. 

This paper is organized as follows. 
In Section~2, we explain the setting of this model and give some related works. 
In Section~3, an explicit expression for the stationary state is computed using the Chebyshev polynomials. From this expression, we obtain the transmitting and reflecting rates and a neseccory and sufficient condition for the perfect transmitting.  We also gives the energy in the perturbed region. 
In Section~4, we estimate the asymptotics of the energy to see  the discontinuity with respect to the incident inflow. 

\section{The setting of our quantum walk}
The total Hilbert space is denoted by $\mathcal{H}:=\ell^2(\mathbb{Z};\mathbb{C}^2)\cong \ell^2(A)$. 
Here $A$ is the set of arcs of one-dimensional lattice whose elements are labeled by $\{(x;R),(x;L) \;|\; x\in \mathbb{Z}\}$, 
where $(x;R)$ and $(x;L)$ represents the arcs ``from $x-1$ to $x$", and ``from $x+1$ to $x$", respectively. 
We assign a $2\times 2$ unitary matrix to each $x\in \mathbb{Z}$ so called local quantum coin 
	\[ C_x=\begin{bmatrix} a_x & b_x \\ c_x & d_x \end{bmatrix}. \]
Putting $|L\rangle := [1,0]^\top$, $|R\rangle:=[0,1]^\top$ and $\langle L|=[1,0]$, $\langle R|=[0,1]$, 
we define the following matrix valued weights associated with moving to left and right from $x$ by 
	\[ P_x=|L\rangle \langle L| C_x,\;\;Q_x=|R\rangle \langle R| C_x, \]
respectively. Then the time evolution operator on $\ell^2(\mathbb{Z};\mathbb{C}^2)$ is described by 
	\[ (U\psi)(x)=P_{x+1}\psi(x+1)+Q_{x-1}\psi(x-1) \]
for any $\psi\in \ell^2(\mathbb{Z};\mathbb{C}^2)$. 
Its equivalent expression on $\ell^2(A)$ is described by 
	\begin{align}\label{def:U'} 
        (U'\phi)(x;L) &= a_{x+1}\phi(x+1;L)+b_{x+1}\phi(x+1;R), \notag \\
        (U'\phi)(x;R) &= c_{x-1}\phi(x-1;L)+d_{x-1}\phi(x-1;R) 
        \end{align}
for any $\psi\in \ell^2(A)$. 
We call $a_x$ and $d_x$ the transmitting amplitudes, and 
$b_x$ and $c_x$ the reflection amplitudes at $x$, respectively
\footnote{If we put $a_x=d_x=1$ and $b_x=c_x=\sqrt{-1}=i$, then the primitive form of QW in \cite{FH} is reproduced.}.
Remark that $U$ and $U'$ are unitarily equivalent such that 
letting $\eta:\ell^2(\mathbb{Z};\mathbb{C}^2)\to \ell^2(A)$ be 
	\begin{align*}
        (\eta \psi)(x;R)= \langle R|\psi \rangle,\;\;(\eta \psi)(x;L)= \langle L|\psi \rangle
	\end{align*} 
then we have $U=\eta^{-1}U'\eta$. 
The free quantum walk is the quantum walk where all local quantum coins are described by the identity matrix i.e.
$$
(U_0 \psi )(x)= \left[ \begin{array}{cc} 1 & 0 \\ 0 & 0 \end{array} \right] \psi (x+1) + \left[ \begin{array}{cc} 0 & 0 \\ 0 & 1 \end{array} \right] \psi (x-1).
$$
Then the walker runs through one-dimensional lattices without any reflections in the free case. 

In this paper we set ``impurities" on 
\[ \Gamma_M:=\{0,1,\dots,M-1\} \] 
in the free quantum walk on one-dimensional lattice; that is, 
	\begin{equation}\label{inpurity} 
        C_x=\begin{cases} \begin{bmatrix}a & b \\ c & d\end{bmatrix} & \text{: $x\in \Gamma_M$,}\\ \\ I_2 & \text{: $x\notin \Gamma_M$.} \end{cases} 
        \end{equation}
We consider the initial state $\Psi_0$ as follows. 
\begin{equation*}
    \Psi_0(x) =  
    \begin{cases}
    e^{i\xi x}|R\rangle & \text{: $x\leq 0$;} \\
    0 & \text{: otherwise,}
    \end{cases} 
\end{equation*}
where $\xi\in \mathbb{R}/2\pi\mathbb{Z}$. 
Note that this initial state belongs to no longer $\ell^2$ category. 
The region $\Gamma_M$ is obtain a time dependent  inflow $e^{-i\xi n}$ from the negative outside. On the other hand, if a quantum walker goes out side of $\Gamma_M$, it never come back again to $\Gamma_M$. We can regard such a quantum walker as an outflow from $\Gamma_M$. Roughly speaking, in the long time limit, the inflow and outflow are balanced and obtain the stationary state with some modification. 
Indeed the following statement holds. 
\begin{proposition}\cite{HS}
\begin{enumerate}
\item This quantum walk converges to a stationary state in the following meaning: 
\[ \exists \lim_{n\to \infty} e^{i(n+1)\xi}\Psi_n(x)=:\Phi_\infty(x).\] 
\item This stationary state is a generalized eigenfunction satisfying 
\[ U\Phi_\infty=e^{-i\xi}\Phi_\infty. \]
\end{enumerate}
\end{proposition}

\noindent{\bf Relation to an absorption problem}\\
Let the reflection amplitude at time $n$ be $\tilde{\gamma}_n(z):= \langle L|\Phi_n(-1)\rangle$ with $z=e^{i\xi}$. 
We can see that  $\tilde{\gamma}_n(z)$ is rewritten by using $U'$ as follows:  
    \begin{multline*}
        z^{-1}\tilde{\gamma}_{n+1}(z) 
        = \langle \delta_{(-1;L)}, U'\delta_{(0;R)}\rangle 
        + \langle \delta_{(-1;L)}, {U'}^2\delta_{(0;R)}\rangle z \\
        +\langle \delta_{(-1;L)}, {U'}^3\delta_{(0;R)}\rangle z^2
        +\cdots + \langle \delta_{(-1;L)}, {U'}^{n+1}\delta_{(0;R)}\rangle z^{n}
    \end{multline*}
The first term is the amplitude that the inflow at time $n$ cannot penetrate into $\Gamma_M$; the $m$-th term is the amplitude that the inflow at time $n-(m-1)$ penetrates into $\Gamma_M$ and escapes $\Gamma_M$ from $0$ side at time $n$. Therefore each term corresponds to the ``absorption" amplitude to $-1$ with the  absorption walls $-1$ and $M$ with the initial state $\delta_{(0;R)}$. Then 
\begin{remark} The reflection amplitude 
$\langle L|\Phi_\infty(-1)\rangle=\lim_{n\to\infty}\tilde{\gamma}_n(z)$ coincides with the generating function of the absorption amplitude to $-1$ with respect to time $n$ 
while the transmitting amplitude $\langle R|\Phi_\infty(M)\rangle=\lim_{n\to\infty}\tilde{\tau}_n(z)$ coincides with the generating function of the absorption amplitude to $M$ with respect to time $n$.  
\end{remark}
Put $\gamma_n:= |\langle \delta_{(-1;L)}, {U'}^n\delta_{(0;R)}\rangle|^2$ and $\tau_n:=|\langle \delta_{(M;R)}, {U'}^n\delta_{(0;R)}\rangle|^2$ which are the absorption/ first hitting probabilities at positions $-1$ and $M$, respectively starting from $(0:R)$.   
From the above observation, for example, we can express the $m$-th moments of the absorption/hitting times to $-1$ and $M$ as follows:  
\[ \sum_{n\geq 1} n^m \gamma_n = \int_{0}^{2\pi}   \overline{\langle L|\Phi_\infty(-1)\rangle}
\left(-i \frac{\partial}{\partial \xi}\right)^m
\langle L|\Phi_\infty(-1)\rangle \frac{d\xi}{2\pi}, \]
\[ \sum_{n\geq 1} n^m \tau_n = \int_{0}^{2\pi}   \overline{\langle R|\Phi_\infty(M)\rangle}
\left(-i \frac{\partial}{\partial \xi}\right)^m
\langle R|\Phi_\infty(M)\rangle \frac{d\xi}{2\pi}. \]

\noindent{\bf Relation to Scattering of quantum walk}\\
The stationary state $ \Phi _{\infty} $ is the generalized eigenfunction of $U$ in $\ell^{\infty} ( \mathbb{Z};\mathbb{C}^2 )$.
The scattering matrix naturally appears in $\Phi_{\infty}$ (see \cite{KKMS}).
In the time independent scattering theory, the inflow can be considered as the incident ``plane wave", and the impurities cause the scattered wave by transmissions and reflections.
Thus we can see the transmission coefficient and the reflection coefficient in $\Phi_{\infty} (x)$ for $x\in \mathbb{Z}\setminus \Gamma_M $. 
For studies of a general theory of scattering, we also mention the recent work by Tiedra de Aldecoa \cite{Ti}.

\section{Computation of stationary state}
\subsection{Preliminary}
Recall that $|L\rangle$ and $|R\rangle$ represent the standard basis of $\mathbb{C}^2$; 
that is, $|L\rangle=[1,0]^\top$ and $R\rangle=[0,1]^\top$. 
Let $\chi: \ell^2(\mathbb{Z};\mathbb{C}^2)\to \ell^2(\Gamma_M;\mathbb{C}^2)$ be a boundary operator such that 
$(\chi\psi)(a)=\psi(a)$ for any $a\in \{(x;R),(x;L)\;|\;x\in\Gamma_M\}$. 
Here the adjoint $\chi^*: \ell^2(\Gamma_M;\mathbb{C}^2)\to \ell^2(\mathbb{Z};\mathbb{C}^2)$ is described by 
	\[ (\chi^*\varphi)(a)=\begin{cases} \varphi(a) & \text{: $a\in \{(x;R),(x;L)\;|\;x\in\Gamma_M\}$,} \\ 0 & \text{: otherwise.} \end{cases} \]
We put the principal submatrix of $U$ with respect to the impurities by $E_M:=\chi U\chi^{*}$. 
The matrix form of $E_M$ with the computational basis $\chi\delta_0|L\rangle, \chi\delta_0|R\rangle,\dots, \chi\delta_{M-1}|L\rangle, \chi\delta_{M-1}|R\rangle$ 
is expressed by the following $2M\times 2M$ matrix: 
	\begin{equation}\label{eq:E_MMat}
	     E_M= \begin{bmatrix} 
        0 & P &        &        &   \\ 
        Q & 0 & P      &        &   \\
          & Q & 0      & \ddots &   \\
          &   & \ddots & \ddots & P \\
          &   &        & Q      & 0 
        \end{bmatrix} 
	\end{equation}
We express the $((x;J),(x';J'))$ element of $E_M$ by 
    \[ (E_M)_{(x;J),(x';J')}:= \bigg\langle \chi\delta_x|J\rangle , E_M \chi \delta_{x'}|J'\rangle \bigg\rangle_{\mathbb{C}^{2M}}. \]

Putting $\psi_n:=\chi \Psi_n$,  we have 
    \begin{align*}
    \psi_{n+1} &= \chi U(\chi^*\chi+(1-\chi^*\chi))\Psi_{n} \\
        &= E_M \psi_{n} +\chi U(1-\chi^*\chi)\Psi_n \\
        &= E_M \psi_n + e^{-i(n+1)\xi} \chi \delta_0|R\rangle. 
    \end{align*}
Then, putting $\phi_n:=e^{i(n+1)\xi}\psi_n$, we have 
    \begin{align}
    e^{-i\xi}\phi_{n+1} &= E_M\phi_n + \chi \delta_0|R\rangle. 
    \end{align}
From \cite{HS}, $\varphi:=\exists \lim_{n\to\infty}\phi_n$. 
Then the stationary state restricted to $\Gamma_M$ satisfies
    \begin{equation}\label{eq:popo}
        (e^{-i\xi}-E_M)\phi_{\infty} =  \chi \delta_0|R\rangle. 
    \end{equation}
%
About the uniqueness of this solution is ensured by the following Lemma because the existence of the inverse of $(e^{-i\xi}-E_M)$ is included by this Lemma. 
\begin{lemma}\label{lemsonzai}
Let $E_{M}$ be the above {with $a\neq 0$}.$^\dagger$ Then 
$\sigma(E_{M})\subset \{ \lambda \in \mathbb{C} \;|\; |\lambda|<1 \}$. 
\end{lemma}
\begin{proof}
Let $\psi\in \ell^2(\Gamma_M,\mathbb{C}^2)$ be an eigenvector of eigenvalue $\lambda\in \sigma(E_M)$. 
Then 
	\begin{align}\label{eq:EM}  
        |\lambda|^2 ||\psi||^2
        	=|| E_M \psi ||^2 = \langle U\chi^*\psi, \chi^*\chi U\chi^*\psi \rangle 
        	\leq \langle  U\chi^*\psi, U\chi^*\psi \rangle 
                =||\chi^*\psi ||^2
                = ||\psi||^2. 
        \end{align}
Here for the inequality, we used the fact that $\chi^*\chi$ is the projection operator onto 
\[ \mathrm{span}\{ \delta_x |L\rangle, \delta_x|R\rangle \;|\; x\in \Gamma_M \}\subset \ell^2(\mathbb{Z};\mathbb{C}^2)\] 
while 
for the final equality, we used the fact that $\chi\chi^*$ is the identity operator on $\ell^2(\Gamma_M;\mathbb{C}^2)$. 
If the equality in (\ref{eq:EM}) holds, then $\chi^*\chi U\chi^*\psi=U\chi^*\psi$ holds. 
Then we have the eigenequation $U\chi^*\psi=\lambda \chi^*\psi$ by taking $\chi^*$ to both sides of 
the original eigenequation $\chi U\chi^*\psi=\lambda \psi$. 
However there are no eigenvectors having finite supports in a position independent quantum walk on $\mathbb{Z}$ with $a\neq 0$ since 
its spectrum is described by only a continuous spectrum in general. 
Thus $|\lambda|^2 < 1$.  
\end{proof}
Now let us solve this equation (\ref{eq:popo}). 
The matrix representation of $E_M$ with the permutation of the labeling such that $(x;R)\leftrightarrow (x;L)$ for any $x\in \Gamma_M$ to (\ref{eq:E_MMat}) is  
\[ E_M\cong \left[
    \begin{array}{cc|cc|cc|cc|cc}
    0 & 0 & 0 & 0 &        &   &       & & &\\ 
    0 & 0 & b & a &        &   &       & & &\\ \hline
    d & c & 0 & 0 & 0      & 0 &       & & &\\ 
    0 & 0 & 0 & 0 & b      & a &       & & &\\ \hline
     &  & d & c &  \ddots     &  &\ddots & & &\\ 
     &  & 0 & 0 &       &  &       & & &\\ \hline
      &   &   &   & \ddots &   &\ddots & & 0 & 0 \\ 
      &   &   &   &        &   &       & & b & a \\ \hline
      &   &   &   &        &   &    d &c &  0& 0\\
      &   &   &   &        &   & 0    & 0&  0& 0 
    \end{array}\right]. \]
Then the equation (\ref{eq:popo}) is expressed by 
\[ \left[
    \begin{array}{cc|cc|cc|cc|cc}
    z & 0 & 0 & 0 &        &   &       & & &\\ \hline
    0 & z & -b & -a &        &   &       & & &\\ 
    -d & -c & z & 0 &       &  &       & & &\\ \hline
     &  & 0 & z & -b      & -a &       & & &\\ 
     &  & -d & -c & z & 0 & & & &\\ \hline
     &  &  &  & \ddots &  & \ddots & & &\\ 
      &   &   &   &  &   & & &  &  \\ \hline
      &   &   &   &        &   &    0 & z & -b & -a \\ 
      &   &   &   &        &   &    -d & -c &  z& 0\\ \hline
      &   &   &   &        &   &     &  &  0& z 
    \end{array}\right]
   \left[
   \begin{array}{c}
        \varphi(0;R) \\ \varphi(0;L)  \\ \hline
        \varphi(1;R) \\ \varphi(1;L)  \\ \hline
        \vdots \\ \vdots \\ \hline
        \varphi(M-2;R) \\ \varphi(M-2;L) \\ \hline
        \varphi(M-1;R) \\ \varphi(M-1;L)
   \end{array}
   \right] 
    =\left[
   \begin{array}{c}
        1 \\ 0  \\ \hline
        0 \\ 0  \\ \hline
        \vdots \\ \vdots \\ \hline
        0 \\ 0 \\ \hline
        0 \\ 0
   \end{array}
   \right].  \]
Here we changed the way of blockwise of $E_M$ and we put $z=e^{-i\xi}$.  
Putting 
    \[ A_z :=\begin{bmatrix} 0 & z \\ -d & -c \end{bmatrix}, \;
       B_z :=\begin{bmatrix} -b & -a \\ z & 0 \end{bmatrix}, \]
we have 
    \begin{align}\label{eq:rec}
        \begin{bmatrix} z & 0 \end{bmatrix}\vec{\varphi}(0) &= 1, \;\; 
        A_z \vec{\varphi}(0) + B_z \vec{\varphi}(1) = 0,\;\;
        A_z \vec{\varphi}(1) + B_z \vec{\varphi}(2) = 0,\dots \notag\\
        & \dots, A_z \vec{\varphi}(M-2) + B_z \vec{\varphi}(M-1) = 0,\;\;
        \begin{bmatrix} 0 & z \end{bmatrix}\vec{\varphi}(M-1) = 0,
    \end{align}
where $\vec{\varphi}(x)=[\varphi(x;R),\varphi(x;L)]^\top$ for any $x\in \Gamma_M$. 
The inverse matrix of $B_z$ exists since $z\neq 0$. 
Then we have 
    \begin{align}\label{eq:T}
        \vec{\varphi}(1)=T\vec{\varphi}(0),\;
        \vec{\varphi}(2)=T^2\vec{\varphi}(0),\dots,\vec{\varphi}(M-1)=T^{M-1}\vec{\varphi}(0), 
    \end{align}
where 
    \[T=-B_z^{-1}A_z=\frac{1}{az}\begin{bmatrix} \Delta |a|^2 & -\Delta a\bar{b} \\ -\Delta \bar{a}b & z^2+\Delta |b|^2 \end{bmatrix}. \]
Here $\Delta=\det (P+Q)$. 
For the boundaries, there exists $\kappa$ such that 
    \begin{equation}\label{eq:boundary}
        \vec{\varphi}(0)=\begin{bmatrix} z^{-1} & \kappa \end{bmatrix},\;
        \begin{bmatrix}0 & z\end{bmatrix}\vec{\varphi}(M-1)=0.        
    \end{equation}
By (\ref{eq:T}) and (\ref{eq:boundary}), $\kappa$ satisfies 
    \begin{equation}
        \left\langle \begin{bmatrix} 0 \\ 1  \end{bmatrix}, T^{M-1}\begin{bmatrix} z^{-1} \\ \kappa \end{bmatrix}  \right\rangle=0
    \end{equation}
which is equivalent to  
    \[ \kappa=-\frac{z^{-1}(T^{M-1})_{2,1}}{(T^{M-1})_{2,2}}.  \]

Now the problem is reduced to consider the $n$-th power of $T$ because the eigengector is expressed by $\vec{\varphi}(n)= T^n\vec{\varphi}(0)$. 
Since $T$ is a just $2\times 2$ matrix,  we can prepare the following lemma.
\begin{lemma}\label{lem:nthpower}
Let $A$ be a $2$-dimensional matrix denoted by 
    \[ A=\begin{bmatrix} \alpha & \beta \\ \gamma & \delta \end{bmatrix}.  \]
\begin{enumerate}
    \item $(\alpha-\delta)^2+4\beta\gamma=0$ and $A\neq \epsilon I$ for some $\epsilon$ case.  
    Let $\lambda=(\alpha+\delta)/2$. Then
        \[ A^n=\begin{bmatrix} 
        \lambda^n+\frac{\alpha-\delta}{2}n\lambda^{n-1} & \beta n\lambda^{n-1} \\ 
        \gamma n \lambda^{n-1} & \lambda^n-\frac{\alpha-\delta}{2}n\lambda^{n-1} \end{bmatrix} \]
    \item Otherwise. 
    Let $\zeta_n:=(\det(A)^{1/2})^{n-1}U_{n-1}(\frac{\mathrm{tr}(A)}{2 \det(A)^{1/2}})$ for $n\geq 1$. Then  
        \[ A^n=\begin{bmatrix} \zeta_{n+1}-\delta\zeta_n & \beta\zeta_n \\ \gamma\zeta_n & \zeta_{n+1}-\alpha\zeta_n \end{bmatrix}, \]
where $U_n(\cdot)$ is the Chebyshev polynomial of the second kind. 
\end{enumerate}
\end{lemma}
\begin{remark}
The condition of ``$(\alpha-\delta)^2+4\beta\gamma=0$  and $A\neq \epsilon I$" is a necessary and sufficient condition of the non-diagonalizability of $A$.  
\end{remark}
\begin{remark}
For $A=T$ case, 
the condition of (1) is reduced to 
\[ \omega:=\Delta^{-1/2}z\in \{ \epsilon_1|a|+\epsilon_2i|b| \;|\; \epsilon_1,\epsilon_2\in \{\pm 1\}  \}=:\partial B.\] 
\end{remark}
\begin{remark}
For $A=T$ case, 
the variable of the Chebyshev polynomial of (2) is reduced to 
\[ \mathrm{tr}(T)/(2 \det(T)^{1/2}) = (\omega+\omega^{-1})/(2|a|). \]
Moreover if $\omega=e^{ik}$, the Chebyshev polynomial is described by $U_{-1}(\cdot)=0$, 
    \begin{align*} 
    U_{n}( \cos k/|a|)
    &=\frac{\lambda_+^{n+1}-\lambda_-^{n+1}}{\lambda_+-\lambda_-}\;\;(n\geq 0).  
    \end{align*}
Here $\lambda_{\pm}$ in RHS are the solutions of the quadratic equation 
\[\lambda^2-\frac{2\cos k}{|a|} \lambda+1=0\] 
with $|\lambda_-|\leq |\lambda_+|$.
\end{remark}
\subsection{Transmitting and reflecting rates}
Let us divide the unit circle in the complex plain as follows: 
\begin{align}
    B_{in} = \{ e^{ik} \;|\; |\cos k|<|a| \},\;
    \partial B = \{ e^{ik} \;|\; |\cos k|=|a| \},\;
    B_{out} = \{ e^{ik} \;|\; |\cos k|>|a| \}.
\end{align}
By the unitarity of $\begin{bmatrix} a & b \\ c & d\end{bmatrix}$ and using the Chebyshev recursion; $U_{n+1}(x)=2xU_n(x)-U_{n-1}(x)$, we insert (1) and (2) in Lemma~\ref{lem:nthpower} into (\ref{eq:T}), and we have an explicit expression for the stationary state as follows. 
    \begin{theorem}\label{lem:stationarystate}
Let the stationary state restricted to $\{0,1,\dots,M-1\}$ be $\phi_\infty$ and $\vec{\varphi}(n):=[\phi_\infty(n;R)\;\phi_\infty(n;L)]^\top$. 
Then we have 
    \begin{equation}
        \vec{\varphi}(n)= 
        \begin{cases}\frac{z^{-1}(\alpha \Delta^{-1/2})^{-n}}{\omega \zeta'_M-|a|\zeta'_{M-1}}\begin{bmatrix}\omega \zeta'_{M-n}-|a|\zeta'_{M-n-1} \\ \alpha b \zeta'_{M-n-1}\end{bmatrix} & \text{: $\omega\notin \partial B$} \\
        \\
        \frac{\Delta^{-1/2}\lambda^n}{\epsilon_R|a|+i\epsilon_IM|b|}\begin{bmatrix}\epsilon_R\alpha(\epsilon_R|a|+i\epsilon_I|b|(M-n))\\ b(M-n-1)\end{bmatrix} & \text{: $\omega \in \partial B$}
\end{cases}
    \end{equation}
for $n=0,1,\dots,M-1$, where $\alpha=a/|a|$ and $\zeta'_m=U_{m-1}(\frac{\omega+\omega^{-1}}{2|a|})$ $(m\geq 0)$, $\lambda=\mathrm{sgn}(\epsilon_R)\alpha^{-1}\Delta^{1/2}$. 
Here $\epsilon_R=\mathrm{sgn}(Re(\omega))$ and $\epsilon_I=\mathrm{sgn}(Im(\omega))$
    \end{theorem}
    
Since the transmitting and reflecting rates are computed by 
    \begin{align*} 
    T(\omega) &= \left|\left\langle \begin{bmatrix} 1 \\ 0 \end{bmatrix}, \vec{\varphi}(M-1)  \right\rangle \times d\right|^2,\\
    R(\omega) &= \left|\left\langle \begin{bmatrix} 0 \\ 1 \end{bmatrix}, \vec{\varphi}(0)  \right\rangle \times a+ \left\langle \begin{bmatrix} 1 \\ 0 \end{bmatrix}, \vec{\varphi}(0)  \right\rangle \times b\right|^2,
    \end{align*}
we obtain explicit expressions for them as follows. 
\begin{corollary}\label{prop:tr}
Assume $abcd\neq 0$. 
For any $\omega\in \mathbb{R}/(2\pi\mathbb{Z})$, we have 
\begin{align}
    T(\omega) &= \frac{|a|^2}{|a|^2+|b|^2{\zeta'}^2_{M}} \\
    R(\omega) &= \frac{|b|^2 {\zeta'}^2_M}{|a|^2+|b|^2{\zeta'}^2_M}
\end{align}
\end{corollary}
Note that the unitarity of the time evolution can be confirmed by $T+R=1$.  
By Corollary~\ref{prop:tr}, we can find a necessary and sufficient conditions for the perfect transmitting; that is , $T=1$. 
\begin{corollary}
Assume $abcd\neq 0$.  
Let $\omega=e^{ik}$ with some real value $k$. Then the perfect transmitting happens if and only if 
        \[ \arccos\left(\frac{\cos k}{|a|}\right) \in \left\{ \frac{\ell}{M}\pi \;|\; \ell\in \{0,\pm 1,\dots,\pm(M-1)\} \right\}   . \]
On the other hand, the perfect reflection never occurs. 
\end{corollary}
Remark that if $\omega\notin B_{in}$, then the perfect transmitting never happens. 

\subsection{Energy in the perturbed region}
Taking the square modulus to $\vec{\varphi}(n)$ in Theorem~\ref{lem:stationarystate}, the relative probability at position $n\in\{0,\dots,M-1\}$ can be computed as follows. 
\begin{proposition}
Assume $abcd\neq 0$. Then the relative probability is described by 
\begin{equation}\label{eq:relativeprob}
    ||\vec{\varphi}(n)||^2
    = \begin{cases}
    \frac{1}{|a|^2+|b|^2{\zeta'}^2_M} \left(|a|^2+|b|^2{\zeta'}^2_{M-n-1}+|b|^2 {\zeta'}^2_{M-n}\right) & \text{: $\omega \notin \partial B$}\\
    \\
    \frac{1}{|a|^2+M^2|b|^2} \left\{ |a|^2+|b|^2(M-n)^2+|b|^2(M-n-1)^2 \right\} & \text{: $\omega\in \partial B$}
    \end{cases}
\end{equation}
\end{proposition}
\begin{proof}
Let us consider the case for $\omega \notin \partial B$. 
Using the property of the Chebyshev polynomial, we have ${\zeta}'_{m+1}{\zeta}'_{m-1}={\zeta'}^2_{m}-1$ and $(\omega+\omega^{-1})/|a| \cdot \zeta_m'=\zeta'_{m+1}+\zeta'_{m-1}$. 
It holds that 
    \begin{align*}
        (\omega+\omega^{-1})\zeta_m'\zeta_{m-1}' 
        &= |a| (\zeta'_{m+1}+\zeta'_{m-1}) \zeta'_{m-1} \\
        &= |a|({\zeta'}^2_{m}+{\zeta'}^2_{m-1}-1).
    \end{align*}
Since $\zeta'_m\in \mathbb{R}$, we have 
    \begin{align*} 
    q(m):=| \omega \zeta'_m-|a| \zeta'_{m-1} |^2 
    &= {\zeta'}^2_m+|a|^2{\zeta'}^2_{m-1}-|a|^2 (\omega+\omega^{-1}) \zeta'_{m}\zeta'_{m-1}\\
    &= |b|^2 {\zeta'}^2_{m}+|a|^2, 
    \end{align*}
Then we have 
    \begin{align*}
        ||\vec{\varphi}(n)||^2 &=  
        \frac{1}{q(M)}(q(M-n)+|b|^2{\zeta'}^2_{M-n-1}) \\
        &= \frac{|b|^2{\zeta'}^2_{M-n}+|a|^2+|b|^2 {\zeta'}^2_{M-n-1}}{|b|^2{\zeta'}^2_M+|a|^2}.
    \end{align*}
\end{proof}
Then we can see how much quantum walkers accumulate in the perturbed region $\{0,\dots,M-1\}$ by 
    \[ \mathcal{E}_M(\omega) =:  \sum_{n=0}^{M-1}||\vec{\varphi}(n)||^2.  \]
We call it the energy of quantum walk. 
\begin{corollary}\label{cor:energy}
Let $\mathcal{E}_M(\omega)$ be the above and 
assume $abcd\neq 0$. Then we have 
    \begin{equation}\label{eq:energy1}
    \mathcal{E}_M(\omega) = 
\red{\frac{1}{|a|^2+|b|^2{\zeta'}^2_M} \left\{ M|a|^2 + \frac{|b|^2}{(\lambda_+-\lambda_-)^2} \left( {\zeta'}^2_{M+1}-{\zeta'}^2_{M-1}-4M \right)  \right\}} 
    \end{equation}
In particular, $\mathcal{E}_M(\cdot)$ is continuous at every $\omega_*\in \partial B$ and  
    \[ \mathcal{E}_M(\omega_*)=\frac{1}{3}\; \frac{M}{|a|^2+|b|^2M^2}\left(3|a|^2+|b|^2+2|b|^2M^2\right). \]
\end{corollary}
\begin{proof}
Using the properties of the Chebyshev polynomial for example, $U_n^2-U_{n+1}U_{n-1}=1$, $T_n=(U_{n}-U_{n-2})/2$, we have 
    \[ (\lambda_+^{m-1}+\lambda_-^{m-1})\zeta_M'=2 T_{m-1}U_{m-1}={\zeta'}^2_m-{\zeta'}^2_{m-1}+1. \]
Then we have 
    \begin{align}
    \sum_{n=0}^{m-1} {\zeta_n'}^2 
    &= \sum_{n=0}^{m-1} \left(\frac{\lambda_+^m-\lambda_-^m}{\lambda_+-\lambda_-}\right)^2 \notag \\
    &= \frac{1}{(\lambda_+-\lambda_-)^2} \left\{  (\lambda_+^{m-1}+\lambda_-^{m-1}) \zeta'_m-2m \right\} \notag\\
    &= \frac{1}{(\lambda_+-\lambda-)^2} ({\zeta'}^2_m-{\zeta'}^2_{m-1}-2m+1) \label{eq:sum}
    \end{align}
Then we have 
    \begin{align*}
        \sum_{n=0}^{M-1} ||\vec{\varphi}(n)||^2 
        &= \frac{1}{|a|^2+|b|^2{\zeta'}^2_M} \left( M|a|^2 +|b|^2 \sum_{n=0}^{M-1} {\zeta'}^2_{M-n-1}+{\zeta'}^2_{M-n} \right) \\
        &= \frac{1}{|a|^2+|b|^2{\zeta'}^2_M} \left\{ M|a|^2 + \frac{|b|^2}{(\lambda_+-\lambda_-)^2} \left( {\zeta'}^2_{M+1}-{\zeta'}^2_{M-1}-4M \right)  \right\}
    \end{align*}
Here we used (\ref{eq:sum}) in the last equality. 

If $\omega\in \partial B$, then by directly computation taking summation of (\ref{eq:relativeprob}) over $n\in\{0,1,\dots,{M-1}\}$, we obtain the conclusion. 
Let us see $\mathcal{E}_M(\cdot)$ is continuous at $\partial B$. 
We put $x:= (1/|a|)\cos k$ and $\zeta'_m(x):=\zeta'_m$. Remark that $\omega\to \omega_*$ implies $|x|\to 1$. In the following, we consider $x\to 1$ case. 
The Taylor expansion of $\zeta'_m(x)$ around $x=1$ is 
    \[ \zeta'_m(1-\epsilon) = m-\frac{m}{3}(m^2-1)\epsilon + O(\epsilon^2). \]
The reason for obtaining the expansion until $\epsilon^1$ order is 
    \[ {\zeta'}^2_{M+1}-{\zeta'}^2_{M-1}-4M=O(\epsilon^2).  \]
around $x=1$. 
Note that $(\lambda_+-\lambda_-)^2=4(x^2-1)$. Then 
\[ (\lambda_+-\lambda_-)^2=-8\epsilon+O(\epsilon) \]
around $x=1$. 
Then inserting all of them into (\ref{eq:energy1}), we obtain 
    \[ \lim_{\omega\to \omega_*}\mathcal{E}_M(\omega)=
    \frac{M}{|a|^2+|b|^2M^2} \left( |a|^2+\frac{|b|^2}{3}+\frac{2|b|^2}{3}M^2 \right). \]
\end{proof}

\section{Asymptotics of Energy}
If $\omega\in \partial B$, then by Corollary~\ref{cor:energy}, it is immediately obtained that  
    \begin{equation}\label{eq:partialBlimit} \lim_{M\to\infty}\frac{\mathcal{E}_M(\omega)}{M}=\frac{2}{3}. 
    \end{equation}
    
\noindent 
Let us consider the case of $\omega\in B_{in}\cup B_{out}$ as follows. 
Note that 
\[ \lambda_{\pm} = 
\begin{cases} 
\mathrm{sgn}(\cos k)e^{\pm \theta} & \text{: $\omega\in B_{out}$,}\\ 
e^{\pm i\theta} & \text{: $\omega\in B_{in}$,}
\end{cases} 
\]
where $(1/|a|)\cos k= \cosh \theta$ ($\omega\in B_{out}$), while $(1/|a|)\cos k= \cos \theta$ ($\omega\in B_{in}$) such that $\sin \theta>0$ and $\sinh \theta>0$.  
To observe the asymptotics of $\mathcal{E}_M(\omega)$ for $\omega\notin \partial B$, we rewrite $\mathcal{E}_M(\omega)$ as follows: 
    \begin{align}\label{eq:Binenergy} 
    \mathcal{E}_M(\omega)
    =\begin{cases}
    \red{\frac{1}{|a|^2\sinh^2\theta +|b|^2\sinh^2 M\theta} \left\{ (-|b|^2+|a|^2\sinh^2\theta)M+\frac{|b|^2}{4}\frac{\sinh 2M\theta \sinh 2\theta}{\sinh^2\theta} \right\}} & \text{: $\omega\in B_{out}$} \\
        \\
   \red{\frac{1}{|a|^2\sin^2\theta +|b|^2\sin^2 M\theta} \left\{ (|b|^2+|a|^2\sin^2\theta)M-\frac{|b|^2}{4}\frac{\sin 2M\theta \sin 2\theta}{\sin^2\theta} \right\}} & \text{: $\omega\in B_{in}$} \\
    \end{cases}
    \end{align}

In the following, we regard $\mathcal{E}_M(\omega)$ as a function of $\theta$, $M$; that is $\mathcal{E}(M,\theta)$ because $\theta$ can be expressed by $\omega$ and consider the asymptotics for large $M$. 

\subsection{$\omega\in B_{out}$}
Let us see that 
\begin{equation}\label{case1Bout} 
\lim_{M\to\infty}\mathcal{E}_M(\omega)= \frac{\cosh \theta}{\sinh \theta}=\frac{\left|\frac{\cos k}{a}\right|}{\sqrt{\left|\frac{\cos k}{a}\right|^2-1}}. 
\end{equation}
Note that $\sinh M\theta\sim e^{M\theta}/2\gg M$. Then by (\ref{eq:Binenergy}), we have 
\begin{align*}
    \mathcal{E}_M(\omega) &\sim \frac{1}{|b|^2 e^{2M\theta}} \times \frac{|b|^2}{4} \frac{e^{2M\theta} \sinh 2\theta}{\sinh^2 \theta}
    = \frac{\cosh \theta}{\sinh \theta}. 
\end{align*}
By (\ref{case1Bout}), if $\omega\to \omega_*\in \partial B$, then $\mathcal{E}_M(\omega)\sim 1/\theta \to \infty$. 
To connect it to the limit for the case of $\omega_*\in\partial B$ described by (\ref{eq:partialBlimit}) continuously, we consider $M\to \infty$ and $\theta\to 0$ simultaneously, so that $M\theta \sim \theta_*\in (0,\infty)$. 
Let us see that 
\begin{equation}
    \mathcal{E}_M(\omega)
    \sim \frac{1}{\sinh^2\theta_*}\left( \red{-1}+\frac{\sinh 2\theta_*}{2\theta_*} \right)M
\end{equation}
Noting that $\sinh m \theta=\sinh m\theta_*\neq 0$, for $m=1,2$ and $\sinh\theta \sim \theta_*/M$, we have 
    \begin{align*}
       \mathcal{E}_M(\omega)
       &\sim \frac{1}{|b^2|\sinh^2 \theta_*} \left\{ \red{-}|b|^2M+\frac{|b|^2}{4}\red{ \frac{\sinh 2\theta_* \times (2\theta_*/M)}{(\theta_*/M)^2} } \right\} \\
        &= \frac{1}{\sinh^2\theta_*}\left( \red{-1}+\frac{\sinh 2\theta_*}{2\theta_*} \right)M 
    \end{align*}
Therefore if we design the parameter $\theta_*$  so that  
    \begin{equation}  \frac{2}{3}=\frac{1}{\sinh^2\theta_*}\left( \red{-1}+\frac{\sinh 2\theta_*}{2\theta_*} \right), 
    \end{equation}
then the energy of $B_{out}$ continuously closes to that of $\partial B$ in the sufficient large system size $M$.

\subsection{$\omega\in B_{in}$}
In this paper, since we determine $\theta$ satisfying $\sin \theta>0$, we set $\theta\in (0,\pi)$. 
Remark that $\mathcal{E}_M(\omega^{-1})=\mathcal{E}_M(\omega)$ for any $\omega\in B_{in}$ because $e^{i\theta}$ is invariant under this deformation. 

By (\ref{eq:Binenergy}), if $\sin\theta\asymp \sin M\theta \asymp 1$, we have 
\begin{equation}\label{eq:Mlarge}
\mathcal{E}_M(\omega)\sim \left(\frac{|a|^2\sin^2\theta+|b|^2}{|a|^2\sin^2\theta+|b^2|\sin^2M\theta}\right)M, 
\end{equation}
for sufficiently large $M$, 
which implies that 
    \begin{equation}\label{eq:case1Bin} 
    M \lesssim \mathcal{E}_M(\omega) \lesssim 
    \red{\left( 1+\frac{|b|^2}{\sin^2 \theta} \right)} M 
    \end{equation}
if $\theta\notin \{0,\pi\}$ is fixed. 
Then we conclude that $\mathcal{E}_M(\omega)=O(M)$ if $\theta\notin \mathbb{Z}\pi$ is fixed for $\omega\in B_{in}$. 
On the other hand, if we design $\theta$ so that the condition of the perfect transmitting is satisfied; $\theta=\pi\ell /M$, $|\ell|\in\{1,\dots,M-1\}$ (see Corollary~\ref{prop:tr}) and choose $\ell$ which is very close to $0$ or $M$, then $|\sin\theta|\ll 1$. 
Note that if $|\sin \theta|\to 0$, which means $\omega\to \omega_*\in\partial B$,  then the coefficient of the upper bound in (\ref{eq:case1Bin}) diverges. 

Then from now on, let us consider the following three cases having a magnitude relation between $\theta$ and $M$; 
\begin{center} 
(i) $1\ll M\ll 1/\sin\theta$;\;\; (ii) $M \asymp 1/\sin \theta$;\;\; (iii) $1/\sin\theta \ll M$. \end{center}

\begin{enumerate}
\item Case (i): $1\ll M\ll 1/\sin\theta$\\
Let us start to evaluate RHS of (\ref{eq:Binenergy}). Since 
    \[ \frac{\sin 2M\theta \sin 2\theta}{4\sin^2 \theta} \sim M\left\{ 1-\frac{1}{3}(1+2M^2)\theta \right\}, \]
the ``$\{\;\}$" part in RHS of (\ref{eq:Binenergy}) can be evaluated by $2|b|^2M^3\theta^2/3$. The denominator of (\ref{eq:Binenergy}) is evaluated by $1/(|b|^2M^2\theta^2)$. Combining them, we have 
\red{\begin{equation} \mathcal{E}_M(\omega)\sim \frac{2M}{3} \end{equation}}
This is consistent with (\ref{eq:partialBlimit}). 

\item Case (ii): $M\asymp 1/|\sin\theta|$\\
Under this condition, the parameter $\theta$ lives around $0$ or $\pi$ if $M$ is large. Since we consider $\theta\in(0,\pi)$, we can evaluate $\sin\theta$ by $\sin\theta\sim  \theta$, or $\sin\theta\sim  (\pi-\theta)$ for large $M$.  
We define $\theta'=\theta$ if $0<\theta<\pi/2$ and  $\theta'=\pi-\theta$ if $\pi/2\leq \theta<\pi$. 
Because $M\sin\theta \asymp 1$ by the assumption, we have $M\theta'\asymp 1$. 
So we put $M\theta'=\theta_*+\epsilon$ with $\theta_*\asymp 1$ and $|\epsilon|\ll 1$. Then up to the value $\theta_*$, let us see   
\begin{equation} \mathcal{E}_M(\omega)
\sim  \begin{cases} \frac{1}{\sin^2\theta_*}\left(1\red{-\frac{\sin 2\theta_*}{2\theta_*}}\right) M & \text{: $\theta_*\notin \mathbb{Z}\pi$,} \\
\frac{|b|^2}{|a|^2\theta_*^2}M^3 & \text{: $\theta_*\in \mathbb{Z}\pi$ and $\epsilon M \ll 1$} \\
\frac{M}{\epsilon^2} & \text{: $\theta_*\in\mathbb{Z}\pi$ and  $\epsilon M\gg 1$}
\end{cases}
\end{equation}
Note that if $\theta_*\notin \mathbb{Z}\pi$, then
$\sin \theta=\sin \theta' \sim \theta_*/M$ and $\sin^2 M\theta=\sin^2 M\theta' \sim \sin^2 \theta_*\neq 0$, $\sin 2M \theta= \sin 2M\theta'\sim \sin 2\theta_*$ and so on. 
Inserting them into (\ref{eq:Binenergy}), we have 
\begin{align*} 
\mathcal{E}_M(\omega)
& \sim \frac{1}{|a|^2\theta_*^2/M^2+|b|^2\sin^2 \theta_*}\left\{ (|a|^2\theta_*^2/M^2+|b|^2)M\red{-\frac{|b|^2}{4} \frac{\sin 2\theta_*\cdot 2\theta_*/M}{\theta_*^2/M^2} }  \right\} \\ 
& \sim \frac{1}{\sin^2\theta_*}\left( 1- \frac{\sin 2\theta_*}{2\theta_*}  \right)M
\end{align*}
On the other hand, if $\theta_*\in \mathbb{Z}\pi$, since $\sin \theta\sim \theta_*/M$ and $\sin M\theta_* \sim \epsilon$, by (\ref{eq:Binenergy}), we have 
\begin{align*} 
\mathcal{E}_M(\omega) 
&\sim \frac{1}{|a|^2\theta^2+|b|^2\epsilon^2} \left\{ |b|^2M-\red{\frac{|b|^2}{4} \frac{2\epsilon \cdot 2 \theta_*/M}{(\theta_*/M)^2}} \right\}\\
& \sim \frac{|b|^2M}{|a|^2{\theta'}^2+|b|^2\epsilon^2}  \\
& \sim 
\begin{cases}
\frac{|b|^2}{|a|^2\theta_*^2}M^3 & \text{: $\epsilon\ll \theta_*/M$}\\
M/\epsilon^2 & \text{: $\epsilon\gg \theta_*/M$}
\end{cases}
\end{align*}

\item Case (iii): $ 1/|\sin\theta| \ll M$ \\
The ``$\{\;\;\}$" part in (\ref{eq:Binenergy}) is estimated by $(|b|^2+|a|^2\sin^2\theta)M$ because $M\theta \gg 1$. 
Then we have 
\begin{equation}\label{eq:caseiii} 
\mathcal{E}_M(\omega)\sim \left(\frac{|a|^2\sin^2\theta+|b|^2}{|a|^2\sin^2\theta+|b^2|\sin^2M\theta}\right)M, 
\end{equation}
for sufficiently large $M$ which is the same as (\ref{eq:Mlarge}). 
Let us consider the following case study:
\[ \text{(a)}\; \max\{|\sin \theta|,|\sin M\theta|\} \asymp 1;\;\;\text{(b)}\; |\sin\theta|,|\sin M\theta|\ll 1. \]
\begin{enumerate}
\item  Let us see $\mathcal{E}_M(\omega)=O(M)$ in this case.  
If $\sin \theta\asymp \sin \theta M \asymp 1$, then the coefficient of $M$ in (\ref{eq:caseiii}) is a finite value, then we have (\ref{eq:case1Bin}).  
On the other hand, if each of $\sin \theta$ or $\sin M\theta\ll1$, then (\ref{eq:caseiii}) implies 
\begin{align}
\mathcal{E}_M(\omega)\sim
\begin{cases}
\frac{1}{\sin^2 M\theta}M & \text{: $\sin \theta\ll \sin M\theta \asymp 1$} \\
(1+\frac{|b|^2}{|a|^2\sin^2 \theta})M & \text{: $\sin M\theta\ll \sin \theta\asymp 1$}
\end{cases}
\end{align}
\item Since $|\sin M\theta|\ll 1$, we evaluate $|\sin M\theta|$ by  
    \[ |\sin M\theta|\sim \min\{ |M\theta|,\;|\pi-M\theta|,\dots,|M\pi -M\theta| \}=:\delta. \]
Then there exists a natural number $m$ such that $|\theta-m\pi/M|=\delta/M$. 
Note that $|\sin \theta|$ is also sufficiently small. Then the natural number $m$ must be $m/M\ll 1$ if $0<\theta<\pi/2$ and $(M-m)/M \ll 1$ if $\pi/2 \leq \theta <\pi$.  
Putting $m':=\min\{m,M-m\}$, we have 
\[ |\sin\theta| \sim |\frac{m'}{M}\pi \pm \frac{\delta}{M}|\sim \frac{\delta}{M}. \]
Therefore $|\sin\theta|\ll |\sin M\theta |\ll 1$ holds. 
Then (\ref{eq:caseiii}) implies 
    \[ \mathcal{E}_M(\omega)\sim \frac{M}{\delta^2}. \]
\end{enumerate}

\end{enumerate}

We summarize the above statements in the following theorem by setting $\theta=O(1/M)$, $\epsilon=1/M^\alpha$ as a special but natural design of the parameters. 
\begin{theorem}
Let us set $\omega\in B_{in}$ so that 
\[ \theta=\theta(M)=\left(x\pi+\frac{1}{M^\alpha}\right)\frac{1}{M} \]
with the parameters $x\in(0,M)\subset \mathbb{R}$ and $\alpha\geq 0$. 
If $x \to 0$ or $x\to M$ with fixed $M$, then  $\mathcal{E}_M(\omega)=O(M)$.  
On the other hand, if we take $M\to\infty$ and fix $x'=\min\{x,M-x\}\asymp 1$, then we have 
\[ \mathcal{E}_M(\omega)=\begin{cases} O(M^3) & \text{: $x'$ is natural number and $\alpha\geq 1$,}\\ O(M^{1+2\alpha}) & \text{: $x'$ is natural number and $0\leq \alpha<1$,}\\ O(M) & \text{: otherwise.} \end{cases} \]
\end{theorem}

\begin{table}[h]
 \begin{tabular}{c||c|c|c}
 & $1 \ll M \ll 1/\theta$ & $1 \ll M\asymp 1/\theta$  & $1/\theta\ll M$ \\ \hline\hline 
$\omega\in \partial B$ &  - & - & $O(M)$ \\ \hline
$\omega \in B_{out}$ &  \multicolumn{2}{|c|}{$O(M)$} & $\begin{cases}O(\theta^{-1}) & \text{: $1/\theta \gg 1$}\\ O(1) & \text{: $1/\theta\asymp 1$}\end{cases}$ \\ \hline 
$\omega \in B_{in}$ & $O(M)$  & 
\multicolumn{2}{c}{\small $\begin{cases} 
O(M^3/{\theta_*'}^2) & \text{: $\theta_*\in \mathbb{Z}\pi$, $\epsilon M\ll 1$}\\
O(M\epsilon^{-2}) & \text{: $\theta_*\in \mathbb{Z}\pi$, $\epsilon M\gg 1$}\\
O(M) & \text{: $\theta_*\notin \mathbb{Z}\pi$} 
\end{cases}$} 
 \end{tabular}
\caption{Asymptotics of the energy of $\mathcal{E}_M(\omega)$: $\cos\theta=(\omega+\omega^{-1})/(2|a|)$, $M\theta=\theta_*+\epsilon$. }
 \label{table:OrderEnergy}
\end{table}

\noindent\\
\noindent {\bf Acknowledgments}
H. M. was supported by the grant-in-aid for young scientists No.~16K17630, JSPS. 
E.S. acknowledges financial supports from the Grant-in-Aid of
Scientific Research (C) Japan Society for the Promotion of Science (Grant No.~19K03616) and Research Origin for Dressed Photon.



\begin{small}
\bibliographystyle{jplain}

\end{small}

\end{document}